\documentclass[11pt]{article}

\AtBeginDocument{%
  }

 \usepackage[T1]{fontenc}

\usepackage{fullpage}
\usepackage[margin=1in, centering]{geometry}
\usepackage[utf8]{inputenc}
\usepackage{dsfont}
\usepackage{graphicx}
\usepackage{amsmath} 
\usepackage{amssymb} 
\usepackage{amsthm} 
\usepackage{xcolor}
\usepackage{hyperref}
\usepackage{enumitem}
\usepackage{thm-restate}
\usepackage{comment}
\usepackage{algorithm2e}
\usepackage{algpseudocode}
\usepackage{mdframed} 
\usepackage{multirow}
\usepackage{makecell}
\usepackage{float}
\usepackage{color, colortbl}
\usepackage{tikz}
  \def\firstcircle{(90:1.75cm) circle (2.5cm)}
  \def\secondcircle{(210:1.75cm) circle (2.5cm)}
  \def\thirdcircle{(330:1.75cm) circle (2.5cm)}
\usetikzlibrary{backgrounds,shadows.blur,fit,decorations.pathreplacing,shapes.misc}
\usetikzlibrary{shapes,arrows,chains,calc}
\newtheorem{theorem}{Theorem}[section]
\newtheorem{lemma}[theorem]{Lemma}
\newtheorem{definition}[theorem]{Definition}

\newtheorem{corollary}[theorem]{Corollary}

\newtheorem{remark}[theorem]{Remark}

\newcommand*{\bra}[1]{\langle#1|}
\newcommand*{\ket}[1]{|#1\rangle}
\newcommand*{\opro}[2]{|#1\rangle\langle#2|}

\DeclareMathOperator{\Tr}{Tr}

\newcommand{\E}{\mathop{\mathbb{E}}} 
\newcommand{\A}{\mathcal{A}} 

\newcommand{\calF}{\mathcal{F}}

\newcommand{\pmz}{\pmb{z}}
\newcommand{\prim}[1]{\textsf{#1}} %
\newcommand{\prs}{\prim{PRS}} %
\newcommand{\efi}{\prim{EFI}} %
\newcommand{\efid}{\prim{EFID}} %
\newcommand{\qefid}{\prim{QEFID}}
\newcommand{\efig}{G} %
\newcommand{\bit}{\ensuremath{\{0,1\}}} %
\newcommand{\qvqa}{\prim{QVQA}} %
\newcommand{\uqvqa}{\prim{UQVQA}} %
\newcommand{\dqvqa}{\prim{DQVQA}} %
\newcommand{\vqa}{\prim{VQA}} %

\newcommand{\dvqa}{\prim{DVQA}}
\newcommand{\uvqa}{\prim{UVQA}}
\newcommand{\circf}{\mathfrak{C}} 
\newcommand{\circfgood}{\circf_{\text{good}}} 

\newcommand{\veps}{\ensuremath{\varepsilon}}

\newcommand{\cdistr}{{D}} 
\newcommand{\distrfam}{\mathfrak{D}} 
\newcommand{\cdistrs}{s} 
\newcommand{\sampler}{\mathcal{S}} 
\newcommand{\sample}{\pmz} 
\newcommand{\cnot}{\mathbf{CNOT}} 
\newcommand{\id}{\ensuremath{\mathbf{1}}} %

\newcommand{\setup}{\mathsf{Setup}}

\newcommand{\pp}{\mathsf{pp}}
\newcommand{\vk}{\mathsf{vk}}

\newcommand{\simulate}{\mathsf{Sim}}
\newcommand{\poq}{\prim{PoQ}}

\newcommand{\sevqa}{\prim{SE-VQA}}
\newcommand{\seuvqa}{\prim{SE-UVQA}}
\newcommand{\seqvqa}{\prim{SE-QVQA}}
\newcommand{\seuqvqa}{\prim{SE-UQVQA}}
\newcommand{\sampmcsp}{\prim{SampMCSP}}
\newcommand{\obsampmcsp}{\prim{ObSampMCSP}}
\newcommand{\mcsp}{\prim{MCSP}}
\DeclareMathOperator{\poly}{\mathsf{poly}}

\DeclareMathOperator{\negl}{\mathsf{negl}}
\newcommand{\class}[1]{\mathsf{#1}}

\usepackage{mathtools}

\DeclarePairedDelimiter\abs{\lvert}{\rvert}
\DeclarePairedDelimiter\norm{\lVert}{\rVert}

\newcommand{\secpar}{\lambda} %
\renewcommand{\S}{\mathcal{S}} %
\renewcommand{\H}{\mathcal{H}} %
\newcommand{\K}{\mathcal{K}} %

\usepackage{cleveref}

\crefname{lemma}{Lemma}{Lemmas}
\crefname{proposition}{Proposition}{Propositions}
\crefname{definition}{Definition}{Definitions}
\crefname{theorem}{Theorem}{Theorems}
\crefname{conjecture}{Conjecture}{Conjectures}
\crefname{corollary}{Corollary}{Corollaries}
\crefname{claim}{Claim}{Claims}
\crefname{section}{Section}{Sections}
\crefname{appendix}{Appendix}{Appendices}
\crefname{figure}{Fig.}{Figs.}
\crefname{table}{Table}{Tables}

\title{A Cryptographic Perspective on the Verifiability of Quantum Advantage}

\author{Nai-hui Chia 
\thanks{Rice University, USA. Email:\href{mailto:nc67@rice.edu}{nc67@rice.edu}}
\and Honghao Fu 
\thanks{MIT, USA, and Concordia University, Canada. Email:\href{mailto:honghao.fu@concordia.ca}{honghao.fu@concordia.ca}}
\and Fang Song 
\thanks{Portland State University, USA. Email:\href{mailto:fang.song@pdx.edu}{fang.song@pdx.edu}}
\and Penghui Yao
\thanks{Nanjing University and Hefei National Laboratory, China. \href{mailto:phyao1985@gmail.com}{phyao1985@gmail.com}}}
\date{}

\begin{document}

\maketitle
\begin{abstract}
  In recent years, achieving verifiable quantum advantage on a NISQ device has emerged as an important open problem in quantum information.
The sampling-based quantum advantages are not known to have efficient verification methods. This paper investigates the verification of quantum advantage from a cryptographic perspective.  We establish a strong connection between the verifiability of quantum advantage and cryptographic and complexity primitives, including efficiently samplable, statistically far but computationally indistinguishable pairs of (mixed) quantum states (\prim{EFI}), pseudorandom states (\prim{PRS}), and variants of minimum circuit size problems ($\mcsp$). Specifically, we prove that a) a sampling-based quantum advantage is either verifiable or can be used to build \prim{EFI} and even \prim{PRS} and b) polynomial-time algorithms for a variant of $\mcsp$ would imply efficient verification of quantum advantages. 
Our work shows that the quest for verifiable quantum advantages may lead to applications of quantum cryptography, and the construction of quantum primitives can provide new insights into the verifiability of quantum advantages.
\end{abstract}

\section{Introduction}

Quantum advantage experiments aim to demonstrate tasks that quantum
computers outperform classical computers. In recent years, random
circuit sampling (RCS)~\cite{boixo2018characterizing} and Boson
sampling~\cite{boson} emerge as promising proposals since they can be
implemented on a NISQ (Noisy Intermediate-Scale Quantum) device and
admit \emph{provably} complexity-theoretical evidence for the hardness
on classical computers~\cite{10.5555/3135595.3135617}. Besides these two desirable
criteria for a quantum advantage experiment, another critical
criterion is the ability to \emph{verify} the outcomes from such
experiments, preferably by an efficient classical
computer. Verification for RCS and Boson sampling both turn out to be
challenging. At present, it remains open to demonstrate a quantum
advantage experiment that satisfies all three of these criteria (see
~\cref{fig:vqa} for a summary).
\begin{figure}[H]
    \centering
    \resizebox{.4\textwidth}{!}{%
    \begin{tikzpicture}
        \begin{scope}
        \clip \firstcircle;
        \fill[yellow] \secondcircle;
        \end{scope}
      \begin{scope}
        \clip \secondcircle;
        \fill[green] \thirdcircle;
      \end{scope}
      \begin{scope}
        \clip \firstcircle;
        \fill[cyan] \thirdcircle;
      \end{scope}
      \begin{scope}
          \clip \firstcircle;
          \clip \secondcircle;
          \clip \thirdcircle;
          \fill[magenta] \thirdcircle;
      \end{scope}
      \draw \firstcircle node[text=violet,above,font=\bfseries] {\large NISQable};
      \draw \secondcircle node [text=violet,below left,align=center,font=\bfseries] {\large Classically\\ \large hard};
      \draw \thirdcircle node [text=violet,below right,align=center,font=\bfseries] {\large Efficiently\\ \large verifiable};
      \node[text width=1cm,fill=white,font=\bfseries] at (0,0) {\large Ideal};
      \node[text width=0.9cm, align=center] at (0, -1.5) {PoQ\\Shor...};
        \node[text width=0.9cm, align=center] at (1.2, 0.9) {QAOA \\ VQE...};
        \node[text width=1cm, align=center] at (-1.4, 0.9) {RCS\\Boson...};
      \end{tikzpicture}
      }
      \caption{Scott Aaronson's categorization of quantum advantage
        proposals. Random circuit sampling
        \cite{boixo2018characterizing} and Boson sampling \cite{boson}
        are NISQable and Classically hard. Cryptographic proof of
        quantumness (PoQ)
        \cite{mahadevPoQ,computationalCHSH} and Shor's
        algorithm \cite{shor} are classically hard and efficiently
        verifiable.  QAOA \cite{qaoa} and VQE \cite{vqe} are NISQable
        and efficiently verifiable.}
      \label{fig:vqa}
\end{figure}

About the verifiability of RCS, the linear cross-entropy benchmarking
(XEB) is first proposed as a verification method
\cite{boixo2018characterizing}.  However, XEB is sample efficient but
not computationally efficient, and it can be spoofed
\cite{gao2021limitations,pan2022simulation}.  More generally, a work
by Hangleiter et. al~\cite{hangleiter2019sample} cast a further shadow
on their verifiability.  They show that if the target distribution
anticoncentrates, certifying closeness to the target distribution
requires exponentially many samples, which covers RCS, Boson sampling
and IQP sampling.  This result rules out efficient verification for
the known quantum advantage experiments based on sampling.

What about general quantum sampling experiments?  How do we determine
if such an experiment has an efficient verification method? 
In \cite{francca2022game}, the verification task is modelled as a game
between a quantum party and a classical challenger, which we will
discuss more later.  
However, they
can only show the limitations of the verification methods that
calculate the empirical average of some scoring functions of individual samples in this model.

\subsection{Our results}
In this paper, we investigate the verifiability of sampling-based quantum
advantage experiments via a \emph{cryptographic} perspective.  To this
end, we first put forth formal definitions of
verifiability. 
Subsequently, we study the implication of the hardness of a variant of the minimal circuit size problem ($\mcsp$) on verifiability.
Furthermore, we establish the connection between
verifiability and fundamental quantum cryptographic primitives: $\efi$
(efficiently generated, statistically far, and computationally
indistinguishable states)
\footnote{More specifically, $\qefid$. Since $\qefid$ is $\efi$, we use $\efi$ in this section for simplicity.}
and $\prs$ (pseudorandom states).  
Lastly, we generalize verifiable quantum advantage to capture the verifiability of interactive proof of quantumness. 
We hope that our work will advance the understanding of the verifiability of
quantum advantage experiments and provide insights into the
development of future quantum advantage experiments.

\begin{figure}[H]
   \centering
\begin{tikzpicture}
    \tikzstyle{operator} = [draw,rounded rectangle,fill=white,minimum size=3em] 
    \tikzstyle{device} = [draw,rounded rectangle,fill=white,minimum size=4.5em] 
    \tikzstyle{phase} = [fill,shape=circle,minimum size=5pt,inner sep=0pt]
    \tikzstyle{surround} = [fill=white,thick,draw=black,rounded corners=2mm]
    \tikzset{edge/.style = {->,> = latex'}}
    \node[operator] (v1) at (2, 0) {Verifier with $\circf$};
    \node[operator] (p1) at (0, -2) {Quantum Alice};
    \node[operator] (p2) at (4, -2) {Skeptic Bob};
    
    \path [->] 
        ([xshift=1.7mm]p1.north) 
        edge node [above=2mm, left=0.5mm,align=center] {($C\in\circf$, $\sample_C$)}  ([yshift=-2.3mm]v1.west);
    
    \path [->] 
        ([xshift=-1.7mm]p2.north) 
        edge node [above=1.5mm, right=0.5mm,align=center] {($S_{C}$, $\sample_{\cdistr_C}$)}  ([yshift=-2mm]v1.east);
    \path [->] 
        (p1.east) 
        edge node [above=0.5mm,align=center] {$C$}  (p2.west);

\end{tikzpicture}
 \caption{Verification process for RCS: 
 The verifier publishes a circuit family $\circf$.
Then Alice sends $C \in \circf$ and samples $\sample_C$ obtained from measuring $C\ket{0^n}$ to the verifier,
and sends $C$ to Bob.
Bob
sends the description of the sampler $\sampler_{C}$ for his classically samplable spoofing distribution $\cdistr_C$ which depends on $C$, along with samples $\sample_{\cdistr_C}$ to the verifier.}
      \label{fig:veri}
\end{figure}

The model of the verification process is depicted in \cref{fig:veri}. It consists of three parties: Alice (a quantum advocate and experiment designer), Bob (a
quantum skeptic) and a verifier~\footnote{We came up with this model
  unaware of the two-party game proposed in \cite{francca2022game},
  albeit the two models share some similarities.
  }.
Alice runs the quantum experiment and sends transcripts of her
experiment, including the setup of the experiment apparatus and
outcomes, to the verifier.  
She also tells Bob about her experiment apparatus but not the outcomes.
Bob proposes a classically
samplable distribution that depends on Alice's experiment and is indistinguishable from Alice's
distribution, and sends the description of his sampling algorithm
along with samples of his distribution to the verifier.  The
verifier's goal is to distinguish Alice and Bob's samples, so in the
rest of the paper we also call him the distinguisher.  The
distinguisher takes all the information from Alice and Bob as input.
In the case of RCS, Alice publishes her random circuit $C$ and her
measurement outcome on $C\ket{0^n}$.  Bob proposes a spoofing
algorithm, which may depend on $C$, and sends the description of the algorithm along with his
samples to the distinguisher.

\begin{definition}[Verifiable quantum advantage (Informal)]
    \label{def:informal_vqa}
    Let $\circf$ be a set of polynomial-sized quantum circuits on $n$ qubits.
    We say the experiment that samples a $C \in \circf$ and repeatedly measures the output state in the computational basis achieves \emph{verifiable quantum advantage} 
    if for all $\distrfam = \{\cdistr_C,\sampler_C\}_{C\in \circf}$ where $\sampler_C$ is a time-$\cdistrs$ classical sampler for $\cdistr_C$, there exists a classical polynomial time distinguisher $\A$ such that 
    \begin{align*}
    \E_{C \gets \circf}|\Pr[\A(C,\sampler_{C}, \sample_{C}) = 1] - \Pr[\A_D(C,\sampler_{C}, \sample_{\cdistr_C}) = 1]|\geq 1/\poly(n), 
    \end{align*}
    where $\sample_{C}$ is a polynomial-sized set of samples generated from measuring $C\ket{0^n}$ in the computational basis, and $\sample_{\cdistr_C}$ is a set of samples drawn from $\cdistr_C$. 
\end{definition}

We give several $\vqa$ examples to demonstrate the expressiveness of
our verifiability definition, such as Fourier sampling (e.g., based on
Shor's algorithm and Simon's problem). Note that our distinguisher is
more general than the ones used in the experiments
\cite{boixo2018characterizing}, studied in \cite{francca2022game}, and the verification algorithms of Fourier sampling problems, because our distinguisher can take the spoofing algorithm $\sampler_C$ as input.
For the sampling-based quantum advantage experiments such as RCS, 
their distinguishers are agnostic about how the classical samples are
sampled. Those distinguishers score each sample individually, and make their decisions
based on the average of the scores.  
As pointed out in
\cite{francca2022game}, if the distinguisher knows the spoofing
algorithm of XEB proposed in \cite{pan2022simulation}, the
distinguisher can distinguish the spoofing samples from the quantum
samples.  
For the Fourier sampling problems, the verification algorithms, i.e. the distinguisher, only depend on the problems themselves.
The distinguisher from our definition can capture these cases by ignoring the description of the classical sampler and deciding by some computation on the samples. 
Hence, we define verifiable quantum advantage with respect to
such a more general distinguisher.

\bigskip

\noindent \textbf{Minimum circuit size problem ($\prim{MCSP}$) vs. $\vqa$}. We aim to identify the computational hardness of verifying quantum advantages. One potential approach is finding a problem for which the existence of efficient algorithms would lead to efficient verification, which is similar to the connections between Meta-complexity problems and cryptography. 

Classical meta-complexity problems, which ask to identify specific complexity measures (e.g., circuit complexity) of given Boolean functions, is a fundamental topic in complexity theory. It is worth noting that efficient algorithms for these problems imply that one-way functions do not exist~\cite{kabanets2000circuit,razborov1994natural}. Chia et al.~\cite{chia2021quantum} investigated quantum minimum circuit size problems ($\mathsf{qMCSP}$) by considering the hardness of identifying quantum circuit complexity of functions, states, and unitary matrices. They showed that the existence of efficient algorithms leads to efficient algorithms for breaking all pseudorandom state schemes and post-quantum one-way functions.   

Inspired by the connections between meta-complexity problems and cryptography, we introduce a variant of meta-complexity problems called the \emph{minimum circuit size problems for samples} ($\sampmcsp$), which asks the minimum size of classical samplers that can generate samples indistinguishable from the given samples. This problem is analogous to the state minimum circuit size problem introduced in~\cite{chia2021quantum}, which asks to identify the quantum circuit complexity of given quantum states. We demonstrate that if $\sampmcsp$ can be solved in polynomial time, then a class of quantum advantage experiments can be verified efficiently. 

\bigskip

\noindent \textbf{$\efi$ vs. $\vqa$}. 
 Next, we study the relationships between
verifiability and the quantum cryptographic primitive \prim{EFI}.  \prim{EFI} is a fundamental quantum cryptographic
primitive, which is equivalent to quantum commitment schemes, quantum
oblivious transfer, quantum multiparty computation, and others
\cite{BCQ23}. 
Note that, classically, one-way functions are necessary but might not
be sufficient to build these applications. 

We show a \emph{duality}
between \prim{EFI} and verifiable quantum advantage. First, we consider
classically secure \prim{EFI} pairs whose computational
indistinguishability holds only against classical algorithms.
\begin{theorem}[Informal]
  \label{thm:informal_efi}
  Suppose that a quantum experiment admits quantum advantage. Then, the experiment is verifiable if and only if there exists a sufficiently large faction of the circuits' output states that do not form an $\efi$ pair with any quantum state that encodes a classical samplable distribution.
\end{theorem}

If we allow a quantum advantage to be verified by a quantum computer, we
obtain a similar duality between quantum-secure \prim{EFI}s and
quantum verifiability. This model with quantum verifiers
is also worth exploring and is discussed in more detail in~\cref{sec:qvqa}.

These results provide necessary and sufficient conditions for
verifiability based on whether the quantum circuit family can form
$\efi$ pairs with classical polynomial-time samplable distributions,
respectively. To the best of our knowledge, all existing $\efi$ pairs
satisfy such a property, i.e., one of the $\efi$ generators can be
simulated by classical polynomial-time algorithms.

\bigskip
\noindent\textbf{Pseudorandom states ($\prs$) vs. $\vqa$}. 
A set of states is \prim{PRS} if a random state in this set is computationally indistinguishable from a Haar random state \cite{JLS18}.
\prim{PRS} is an essential quantum cryptographic primitive that can be used to build other primitives, including one-time digital signatures and \prim{EFI}.
Moreover, the existence of \prim{PRS} implies the existence of
\prim{EFI}, and thus the aforementioned applications that are
equivalent to \prim{EFI} can also be constructed from \prim{PRS}.
Moreover, there is evidence
showing that the existence of \prim{PRS} is a weaker
assumption than the existence of one-way
functions~\cite{kretschmer2021quantum}.

Intuitively, if the output states of a quantum advantage experiment
are pseudorandom, the measurement output distribution should be
indistinguishable from the measurement output distribution of Haar random
states. Moreover, the measurement output distribution of Haar random
states can be approximated by a classical distribution, so this
quantum advantage experiment does not achieve verifiability. However,
in the definition of \prim{PRS}, the distinguisher is unaware of the
preparation circuit of the given state, but the distinguisher in a
quantum advantage experiment is. Hence, we can only prove this result
for a subclass of \prim{PRS}, called classically unidentifiable
\prim{PRS}, which intuitively says that when distinguishing samples
from measuring different states, knowing the circuit does not
help. Many existing \prim{PRS} constructions, such as random phase
states and binary phase states~\cite{JLS18,BS19}, are classically
unidentifiable.

\begin{theorem}[Informal]
    \label{thm:informal_prs}
    If the quantum advantage of a quantum sampling algorithm is verifiable, then the output states are not classically unidentifiable \prim{PRS}. 
\end{theorem}

The motivation behind \cref{thm:informal_prs} is that RCS is proposed as a
candidate construction of $\prs$ \cite{kretschmer2023quantum}.  
If the output states of random circuits are classically unidentifiable,
\cref{thm:informal_prs} gives us a proof that RCS
experiments are unverifiable.  
Note that \cite{hangleiter2019sample} shows that the distribution induced by measuring a random circuit is indistinguishable
from some classical distribution, which does not imply RCS is not $\vqa$ according to \cref{def:informal_vqa}.  
Conversely, \cref{thm:informal_prs} also tells us that if some
construction of $\prs$ fails, it is possible to use this
construction for verifiable quantum advantage.  
This is a win-win situation.

\bigskip
\noindent\textbf{What about interactive quantum advantage experiments?}
So far, we have focused on sampling-based quantum advantage experiments.
There are interactive verifiable quantum advantage proposals called proof of quantumness ($\poq$) \cite{mahadevPoQ,simplePoQ,computationalCHSH}.
These $\poq$s achieve verifiability, but one obstacle in implementing these protocols is maintaining coherence during interactions.

Hence, we generalize \cref{def:informal_vqa} to capture the strength of both \cref{def:informal_vqa} and the verifiability of $\poq$.
In the generalized definition, 
the trusted party is the \emph{designated verifier},
who generates public parameters and a private verification key.
After getting all the samples, the designated verifier uses the verification key to distinguish Alice's quantum samples from Bob's samples.
We call this \emph{Designated verifiable quantum advantage} or $\dvqa$.

Under this definition, the trusted verifier is offline, so Alice does not need to interact with the trusted verifier and can generate the samples on her own as in \cref{def:informal_vqa}.
Moreover, it is possible to compile existing $\poq$ to satisfy the new definition. For example:
Assuming a random oracle, the interactive protocol of \cite{mahadevPoQ} fits this definition. 
The function keys and trapdoors of their protocol are the public parameters and private verification keys here.
Then, the classical or quantum prover can run the operations of the verifier in the original protocol locally by querying the random oracle for the challenges. 
In the end, the prover sends all the generated transcripts to the distinguisher $\A$, who uses the verification key to distinguish the transcripts.
In the compiled protocol, the verifier is offline as in \cref{def:informal_vqa}, and the verifiability of the original $\poq$ is preserved.

\bigskip
\noindent \textbf{Implications.}
We offer a
few perspectives. 

\begin{itemize}
\item For a quantum advocate (experiment designer): The study of quantum cryptography can provide new insights into designing a verifiable quantum advantage experiment. For example, one possible route indicated by \Cref{thm:informal_efi} is to start with a classically insecure \prim{EFI}, and then apply some amplification technique to dilate the statistical distance to obtain the strong quantum advantage while remaining classically insecure. 
\item For a quantum skeptic: 
A spoofing strategy can be found through the lens of quantum cryptography.
\cref{thm:informal_efi} says that the spoofing distribution can be a distribution that forms an \prim{EFI} pair with most of the output states.
\cref{thm:informal_prs} says that if the output states of an experiment are classically unidentifiable \prim{PRS}, the uniform distribution suffices.
\item For a quantum cryptographer: The quest for verifiability of quantum advantages might lead to quantum cryptographic applications.
\cref{thm:informal_efi} implies that if
an experiment is not verifiable, then it will form a classical-secure
\prim{EFI} with a classical polynomial-time samplable distribution.
Since \cref{thm:informal_efi} can be lifted against quantum adversaries,
it is possible to build
standard \prim{EFI} and primitives based on \prim{EFI} from a quantumly unverifiable experiment.
\end{itemize}

In summary, our results show connections between the verifiability of quantum advantages and the quantum cryptographic primitives. It is worth noting that computational tasks demonstrating quantum advantages on near-term quantum devices might not directly result in useful applications; however, our results show that the quest for quantum advantages and their verifiability can provide new insights and methods to build fundamental quantum cryptographic primitives.  

\subsection{Open problems} 
As this is only an initial attempt at
studying the relationship between the verifiability of quantum
advantage experiments and quantum cryptographic primitives, there are
many open problems.  We list some of them here.
\begin{description}
    \item[Random circuits, \prim{PRS}, and \prim{EFI}.] Are the output states of random circuits \prim{PRS}, or even classically unidentifiable \prim{PRS}? Similarly, can we use random circuits to construct \prim{EFI}? 
    There is evidence that the output states of random circuits are \prim{PRS}.
    For example, it is known that polynomial-sized random circuits are approximate poly-designs \cite{BHH}, which indicates that output states of random circuits are highly indistinguishable from Haar random states. Also, it is possible to build an $\efi$ by pairing random circuits and other sufficiently random samplers while ensuring the two output states are statistically far.   
    \item[Quantum cryptography on NISQ devices.] If the output states of random circuits are not \prim{PRS}, can we still construct \prim{PRS} using less structured NISQ circuits with a fixed architecture? The known constructions of \prim{PRS} \cite{JLS18,BS19,BS20,pseudoentangle} all require structured circuits, although some of them only require shallow circuits. 
    If the NISQ device can construct classically \emph{identifiable} \prim{PRS}, our result cannot rule out the possibility of verifiable quantum advantage on NISQ devices. Similarly, it is interesting to know whether one can use NISQ devices to construct $\efi$. 
    \item[The effect of noise.] It is known that efficient sampling from the output distribution of a noisy random quantum circuit can be done classically \cite{aharonov2023polynomial}. 
    What would be the implication of this result on \prim{PRS}?
    If the output states of \emph{noiseless} random circuits are \prim{PRS}, will noisy circuits still output \prim{PRS}? Intuitively, noise will lead to mixed states, which might affect the security of $\prs$ since a Haar random state is a pure state. Along this line, would it be possible to change the definition of $\prs$ to be indistinguishable from ``noisy Haar random states?'' while keeping all the applications of the original $\prs$? Likewise, how would noise impact the construction of $\efi$? Note that $\efi$ must be two computationally indistinguishable and statistically far mixed states. 
    Thus, noise could even make the states more indistinguishable. On the other hand, if the noise is too large, the two states might be statistically close. 
    \item[$\dvqa$: reduced trusted setup and generic compiler.] Our current transformation of existing $\poq$ protocols to a $\dvqa$ experiment is proven in the random oracle model. Can we replace the random oracle with a suitable family of hash functions such as correlation intractable hash \cite{canetti2016ci}? Ideally, can one design a generic compiler that converts any $\poq$ protocol directly to a $\dvqa$ system?
    \item[Complete characterization of verifiability.] In this work, we identify several basic conditions that give useful characterizations of verifiable quantum advantage. 
   It would be fruitful to find other characterizations of verifiable quantum advantage and investigate their applications in quantum information and cryptography. 
\end{description}

\paragraph{Organization.}
We define the quantum primitives in \cref{sec:prelim}. We formally define verifiable quantum advantage and discuss its connection to a variant of $\mcsp$ in \cref{sec:vqa}. 
In \cref{sec:efi}, we explore the relationship between verifiability and $\efi$, and in \cref{sec:prs} we show the relationship between verifiability and $\prs$.
Then we define and discuss $\dvqa$ in \cref{sec:nizk}.
Finally, in \cref{sec:qvqa}, we lift verifiability to against quantum distinguishers and explore its relation with $\efi$.

\paragraph{Acknowledgement.} 
We thank Yunchao Liu for helpful discussions.
NHC was supported by NSF award FET-2243659,
Google Scholar Award, and DOE award DE-SC0024301.
HF was supported by the US National Science Foundation QLCI program (grant OMA-2016245).
FS was supported in part by the US National Science Foundation grants CCF-2054758 (CAREER) and CCF-2224131.
P.Y. were supported by National Natural Science Foundation of China (Grant No. 62332009, 12347104), Innovation Program for Quantum Science and Technology (Grant No. 2021ZD0302901), NSFC/RGC Joint Research Scheme (Grant no. 12461160276) and Natural Science Foundation of Jiangsu Province (No. BK20243060).

\section{Quantum cryptographic primitives}
\label{sec:prelim}

\begin{definition}[$\efi$ pairs~\cite{BCQ23}]
\label{def:efi}
An $\efi$ pair generator is a quantum algorithm
$G: (b,1^{\secpar}) \mapsto \rho_b$ that on inputs $b\in\bit$ and
security parameter $\secpar$, outputs a quantum state $\rho_b$, such
that the following conditions hold.
\begin{itemize}
\item(Efficient preparation) $\efig$ runs in quantum polynomial time.
\item(Statistically far) $\rho_0$ and $\rho_1$ are statistically distinguishable, i.e.,
  $\frac{1}{2}\norm{\rho_0-\rho_1}_1 \geq 1/\poly(\secpar)$.
\item (Indistinguishability) $\rho_0$ and $\rho_1$ are computationally indistinguishable,
  i.e., for all quantum poly-time algorithm $\A$,
  $|\Pr[\A(\rho_0)=1] - \Pr[\A(\rho_1)=1]|\leq \negl(\secpar)$.
\end{itemize}
\end{definition}

There are a few other special cases that are worth noting.
\begin{itemize}
  
\item (\efid) When all objects are specialized to their classical
  counterparts, we recover the classical primitive of $\efid$
  pairs~\cite{Goldreich90}. Namely, $G$ is a poly-time classical
  algorithm which produces samples from one of two distributions
  $D_b$, and the distinguisher $\A$ is an arbitrary classical
  poly-time algorithm. We can view
  $\rho_b:= \sum_{i} D_b(i) \ket{i}\bra{i}$ as a (mixed) state
  encoding $D_b$.

\item (Quantum-secure \efid) If the indistinguishability of $\efid$
  holds against poly-time \emph{quantum} distinguishers and other
  objects remain classical; we call it a quantum-secure \efid.

\item (Classically secure \qefid) If the distributions of $\efid$
can be quantumly generated but the
  indistinguishability of $\efid$ is only required to hold against
  poly-time \emph{classical} distinguishers, we call it a
  classical secure \qefid \cite{qefid}. 
\item (Quantum-secure \qefid) 
This is $\qefid$ whose indistinguishability holds against poly-time \emph{quantum} distinguishers.
Clearly, any quantum-secure $\qefid$ is immediately an $\efi$ by definition \cite{q-qefid}.
\end{itemize}

We show the following fact about \efi.
\begin{lemma}\label{lem:efihybrid}
    Given security parameter $\lambda$, a pair of quantum states $(\rho_0,\rho_1)$ is an \efi~pair if and only if $(\rho_0^{\otimes\poly(\lambda)},\rho_1^{\otimes\poly(\lambda)})$ is an \efi~ pair.
\end{lemma}
\begin{proof}
It suffices to prove the "only if" direction as the other direction is trvial.  It is proved by a standard hybrid argument.
    Let $(C_0, C_1)$ be a pair of $\efi$ generators, which generate $\rho_0$ and $\rho_1$, respectively. We first show that if there exists an algorithm $\A$ such that $|\Pr[\A(\rho_0^{\otimes t(n)})=1] - \Pr[\A(\rho_1^{\otimes t(n)})=1]|> \negl(n)$ for some polynomial $t(\cdot)$, then $\A$ can break the $\efi$ pair $(C_0, C_1)$. 
    
    We prove this by a hybrid argument. Let $H_i = \rho_0^{\otimes i} \otimes \rho_1^{\otimes (t-i)}$ for which $H_0 = \rho_0^{\otimes t}$ and $H_t = \rho_1^{\otimes t}$. Since $\A$ can distinguish $H_0$ from $H_t$, there must exist an $i^*$ for which $\A$ can distinguish $H_{i*}$ from $H_{i^*+1}$. 
    Then, we can construct a distinguisher $\A'$ to distinguish $\rho_0$ from $\rho_1$ as follows: a) $\A'$ first chooses an $i$ uniformly randomly, b) $\A'$ prepares the state $\rho_0^{\otimes i}\otimes \rho \otimes \rho_1^{\otimes t-i-1}$, where $\rho$ is the input state of the $\efi$ game, c) and then $\A'$ runs $\A$ on $\rho_0^{\otimes i} \otimes \rho \otimes \rho_1^{\otimes t-i-1}$. Note that $\rho_0^{\otimes i}\otimes \rho \otimes \rho_1^{\otimes t-i-1}$ is $H_i$ when $\rho= \rho_0$ and is $H_{i+1}$ otherwise. The probability that $\A'$ succeeds is noticeable since $\A'$ chooses $i=i^*$ with probability $1/t$ in a) and $\A$ distinguishes $H_i$ from $H_{i+1}$ with noticeable probability if $i=i^*$. This completes the proof. 

\end{proof}

\begin{definition}[Pseudorandom states ($\prs$)~\cite{JLS18}]
\label{def:prs}
Let $\secpar$ be the security parameter. Let $\H$ be a Hilbert space and
$\K$ be a key space, both parameterized by $\secpar$. A keyed family
  of quantum states $\{\phi_k \in \S(\H) \}_{k\in \K}$ is
  \emph{pseudorandom} if the following hold:

  \begin{enumerate}
  \item (Efficient generation). There is a polynomial-time quantum
    algorithm $G$ that generates state $\ket{\phi_k}$ on input
    $k$. That is, for all $k \in \K, G(k) = \ket{\phi_k}$.
 
  \item (Pseudorandomness). Any polynomially many copies of
    $\ket{\phi_k}$ with the same random $k \gets \K$ is
    computationally indistinguishable from the same number of copies
    of a Haar random state. More precisely, for any efficient quantum
    algorithm $\A$ and any $m \in \poly(\secpar)$,
    \begin{equation*}
      \left| \Pr_{k\gets \K}[\A(\ket{\phi_k}^{\otimes m}) = 1] -
        \Pr_{\phi \gets \mu}[\A(\ket{\phi}^{\otimes m})  = 1\right|\le
      \negl(\secpar) \, , 
    \end{equation*}
where $\mu$ is the Haar measure on $\S(\H)$.
\end{enumerate}
\end{definition}

\paragraph{Verifiability of $\prs$} Let $\ket{\psi}$ be a state generated from $\mathsf{StateGen}$. Given the corresponding key $k$ and $\ket{\psi}^{\otimes \poly(m)}$, one can verify that $\ket{\psi}$ is generated from $\mathsf{StateGen}(k)$ via swap test. It is worth noting that this verification procedure requires implementing the swap test for quantum states.

\begin{definition}[Classically unidentifiable state family]
    \label{def:unidentifiable}
  Let $\lambda$ be the security parameter. Let $\Phi: = \{\ket{\phi_k}\}_{k\in \K}$ be a family of efficiently
  generatable states, i.e., there exists efficient $C_k$ such that
  $\ket{\phi_k} = C_k\ket{0^n}$. We call $\Phi$ \emph{classically
    unidentifiable} if for any efficient classical algorithm $A$, any
  $ i \ne j$, and any polynomial $m = \poly(\secpar)$
\begin{equation*}
  \left| \Pr[A(C_i, \pmz_i) = 1] - \Pr[A(C_i, \pmz_j)=1] \right| \le \negl(\secpar) \, ,
\end{equation*}
where $\pmz_i: = (z_i^1,\ldots,z_i^m)$ are the outcomes by measuring
$\ket{\phi_i}^{\otimes m}$ in the computational basis.
  \label{def:cunid}
\end{definition}

\begin{remark} 
  The random phase state family $\{ \ket{\phi_k}\}$ proposed
  in~\cite{JLS18} below is an example of classically unidentifiable
  $\prs$,
\begin{align*}
  \ket{\phi_k} = \frac{1}{\sqrt{N}}\sum_{x \in [N]} \omega_N^{f_k(x)}
  \ket{x} \, ,
\end{align*}
where $N=2^n$ and $\{ f_k : [N] \to [N] \mid k \in \K\}$ is a
quantum-secure pseudorandom function. This is because when measured in
the computational basis, $\ket{\phi_k}$ always induce the uniform
distribution. Similarly the special case of binary phases~\cite{BS19},
i.e.,
$\ket{\phi_k} = \frac{1}{\sqrt N}\sum_{x \in \bit^n} (-1)^{f_k(x)}
\ket{x}$ is also a classically unidentifiable \prs.
\end{remark}

\section{Verifiable quantum advantages}
\label{sec:vqa}

\begin{definition}[Verifiable quantum advantages (($\cdistrs,t,\veps$)-\vqa)]
\label{def:weak_vqa}
Let $\secpar$ be the security parameter. 
Let $\circf$ be a family of
polynomial-size quantum circuits on $n$ qubits, where $n=\poly(\secpar)$. 
We call $\circf$ a family of $(\cdistrs, t,\veps)$
\emph{verifiable quantum advantage (\vqa)}, if for all $\distrfam = \{\cdistr_C,\sampler_C\}_{C\in \circf}$ where $\sampler_C$ is a time-$\cdistrs$ classical sampler for $\cdistr_C$, the following holds:
for a uniformly random $C$ drawn from
$\circf$ there exists a time-$t$ \textbf{classical} algorithm $\A$ that
$\veps$-distinguishes $t$ samples $\sample_C$ obtained from measuring $C\ket{0^n}$ in the computational basis from
$t$ samples $\sample_{\cdistr_C}$ of $\cdistr_C$, namely
\begin{align*}
  \E_{C\gets \circf}\left| \Pr[\A(C, \sampler_C, \sample_{C}) = 1] - \Pr[\A(C,
  \sampler_C, \sample_{\cdistr_C}) =
  1]\right|\geq \veps \, . 
\end{align*}
\end{definition}

In this definition, the spoofer is aware of the quantum experiment setup. In particular, they know 
which circuit is performed, so their spoofing distribution can depend on that circuit.
However, we don't allow the spoofing distribution to depend on the distinguisher's algorithm because
the distinguisher can choose to publish his algorithm after receiving all the samples.

This work focuses on the case where $s,t= \poly(\secpar)$ and $\epsilon = 1/\poly(\secpar)$. I.e., we ask whether a classical polynomial-time verifier $\A$ can distinguish the quantum samples from the classical samples with noticeable probability. 
However, one can also consider verifiers with different powers by choosing the proper parameters.

The distinguishers considered in the literature of quantum advantage experiments are weaker than ours because their distinguishers are agnostic of the sampler of the classical distribution. Hence, we give an alternative definition of verifiable quantum advantage below.

\begin{definition}
  [Universally verifiable quantum advantages (($\cdistrs, t,\veps$)-\uvqa)]
\label{def:s_vqa}
Let $\secpar$ be the security parameter. 
Let $\circf$ be a family of
polynomial-size quantum circuits on $n$ qubits, where $n=\poly(\secpar)$. 

We call $\circf$ a family of $(\cdistrs, t,\veps)$
\emph{universally verifiable quantum advantage (\uvqa)}, if for \emph{all}
if for all $\distrfam = \{\cdistr_C, \sampler_C\}_{C\in \circf}$ where $\sampler_C$ is a time-$\cdistrs$ classical sampler for $\cdistr_C$, it holds that for a uniformly random $C$ drawn from
$\circf$, there exists a time-$t$ \textbf{classical} algorithm $\A$ that
$\veps$-distinguishes $t$ samples $\sample_C$ obtained from measuring $C\ket{0^n}$ in the computational basis from
$t$ samples $\sample_{\cdistr_C}$ of $\cdistr_C$, namely
\begin{align*}
  \E_{C\gets \circf}\left| \Pr[\A(C, \sample_{C}) = 1] - \Pr[\A(C, \sample_{\cdistr_C}) =
  1]\right|\geq \veps \, .
\end{align*}
\end{definition}
Comparing the two definitions, it is easy to see that for the same set of parameters, if $\circf$ is \uvqa, it is also \vqa,
and if it is not \vqa, it is not \uvqa~either.

\paragraph{Characterizing $\vqa$.} 
We can give an alternative characterization of $\vqa$. 
For a circuit family $\circf$, we define $\circfgood$ to be the subset of circuits that show statistical quantum advantage and their output distributions are efficiently distinguishable from any efficiently samplable classical distributions. 
\begin{definition}[Classically samplable state $\rho_D$]
Let $\mathcal{D} = \{p_{0^n},\dots,p_{1^n}\}$ be some distributions over $\{0,1\}^n$ for which there exists a PPT algorithm that can efficiently sample from $\mathcal{D}$. We define $\rho_\mathcal{D} = \sum_{x\in \{0,1\}^n} p_x\opro{x}{x}$. 
\end{definition}

\begin{definition} 
\label{def:Cgood_c}
We say $C\in \circf$ is good, if for any distribution $\cdistr_C$ samplable by an efficient classical algorithm $\sampler_C$, the two conditions hold for some $\veps = 1/{\poly(\secpar)}$:
\begin{enumerate}[label=\roman*.]
    \item (Quantum advantage) $\norm{C\opro{0^n}{0^n}C^\dagger - \rho_{\cdistr_C}}_1 \geq \veps $.
    \item (Computationally distinguishability) $\exists$ efficient classical $\A$ s.t. 
    \begin{align*}
        \left| \Pr[\A(C, \sampler_C,
    \sample_C)  = 1] - \Pr[\A(C, \sampler_C, \sample_{\cdistr_C})
    = 1]\right| \ge \veps     
    \end{align*}
    where $\sample_C$ are samples from measuring $C\ket{0}^n$ in the computational basis,
    and $\sample_{\cdistr_C}$ are samples of $\cdistr_C$. 
\end{enumerate}
We denote $\circfgood: = \{C\in \circf: \text{$C$ is good} \}$. 
\end{definition}
Then for $\circf$ to be a $\vqa$, $\circfgood$ has to constitute a noticeable fraction and vice versa. 

\begin{lemma} $\circf$ admits $\vqa$ iff. $\frac{\left|\circfgood\right|}{|\circf|} \ge 1/{\poly(\secpar)}$. 
    \label{lemma:vqa=cgood_c}
\end{lemma}

\begin{proof}

First, the ``if'' direction follows directly from our definition of $\vqa$. 

For the ``only if'' direction, suppose that $\circf$ is
  a $\vqa$ family. Then for any efficient $\distrfam = \{\cdistr_C, \sampler_C\}_{C\in \circf}$, there is a poly-time $\A$, such that
  \begin{align*}
    \E_{C\gets \circf} \left| \Pr[\A(C, \sampler_C,
    \sample_{C}) = 1] - \Pr[\A(C, \sampler_C, \sample_{\cdistr_C})
    = 1]\right| \ge \veps \, ,
  \end{align*}
  for some $\veps \ge \frac{1}{\poly(\secpar)}$ using
  $m=\poly(\secpar)$ samples. 
  This implies
  that for any $\distrfam = \{\cdistr_C, \sampler_C\}_{C\in \circf}$ there must exists a $1/\poly(\secpar)$ fraction of $C\in\circf$
  such that $\A$ successfully tells apart $\sample_C$ and $\sample_{\cdistr_C}$ with inverse-poly
  probability, which 
  further implies that the trace distance between $C\opro{0^n}{0^n}C^\dagger$ and $\rho_{\cdistr_C}$ is at least $1/\poly(\secpar)$. 
    
\end{proof}

\subsection{Example: Verifiable quantum advantage}

Here, we cast some sampling problems as verifiable quantum advantage. 
\paragraph{$\vqa$ from quantum Fourier sampling.} 
\begin{definition}[Simon's problem]
Let $f:\{0,1\}^n\rightarrow \{0,1\}^m$ be a 2-to-1 function with the promise that there exists $s\in \{0,1\}^n$ such that $f(x) = f(x \oplus s)$ for all $x\in \{0,1\}^n$. Given oracle access to $f$, find $s$. 
\end{definition}
It is well-known that there exists a quantum polynomial time algorithm solving Simon's problem, and classical algorithms must use superpolynomially many queries \cite{simon1997power}.   

\begin{theorem}
Let $\mathcal{F}$ be the set of all Simon's functions. Then, relative to $\mathcal{F}$, there exists a quantum circuit family $\circf^{\mathcal{F}}$ that is $\uvqa$.  
\end{theorem}

\begin{proof}
    First, we use Simon's algorithm to form the quantum circuit family $\circf^\mathcal{F}$ as follows: Given a random Simon's function $f$ with hidden shift $s\in \{0,1\}^n$, the quantum circuit $C^f$ implements the Simon's algorithm to obtain the quantum state $\rho_f$. Note that when measuring $\rho_f$ in the computational basis, one will obtain a random $x$ for which $x\cdot s = 0$. Hence, our circuit family is defined as $\circf^{\mathcal{F}}:=\{C^f: f\in \mathcal{F}\}$. 

    Our distinguisher $\A^f$ is as follows: On inputs $C^{f}$ and sufficiently many samples $x_0,\dots,x_m$, $\A^f$ runs Gaussian elimination (the classical post-processing in Simon's algorithm) to identify $s$ and then check if $s$ is the hidden shift of $f$. 

    Obviously, samples generated from measuring $C^f$ in the computational basis will be accepted by $\A^f$ with high probability. On the other hand, no efficient classical algorithms can generate samples accepted by $\A^f$ with noticeable probability; this follows from the fact that no polynomial-time classical algorithm can solve Simon's problem with noticeable probability even if the classical algorithm depends on $f$. 
\end{proof}

Following a similar idea, we can obtain the following corollary by considering Shor's algorithm for the Factoring problem. 
\begin{corollary}
    Assuming factoring is hard for any classical polynomial-time algorithm. Then, there exists a quantum circuit family that is $\uvqa$. 
\end{corollary}
Note that neither $\circf^\calF$ nor Shor's algorithm can be implemented on NISQ devices.

\paragraph{Cross-entropy benchmark (XEB).}

In the cross-entropy benchmark (XEB)~\cite{boixo2018characterizing}, given the description of a random quantum circuit $C$, the quantum machine prepares multiple samples $x_1,\dots.x_k\in \{0,1\}^{n}$ accordingly, the verifier tests whether $F_{XEB} = \frac{\sum_{i=1}^k \bra{x_i}C\ket{0^n}}{k}$ is close to $2/2^n$ or close to $1/2^n$. If $F_{XEB}$ is close to $2/2^n$, then the samples are prepared from a quantum machine. If $F_{XEB}$ is close to $1/2^n$, the samples are prepared by some classical machines. 

Suppose that RCS has quantum advantages as described in \cref{def:Cgood_c}, under the linear cross-entropy quantum threshold assumption (XQUATH) \cite{XHOG}, then RCS is a ($s, t,\veps$)-$\uvqa$ with $s = \poly(n)$, $t= \omega(\poly(n))$ and $\veps$ a constant. 
The reason that $t$ is superpolynomial in $n$ is that a classical machine requires time a superpolynomial in $n$ to compute $\abs{\bra{x}C\ket{0^n}}^2$.

Similarly, Boson~\cite{boson} and instantaneous quantum polynomial (IQP)~\cite{iqp} sampling experiments are all \uvqa ~with verification time a superpolynomial in $\secpar$ if the experiments have achieved quantum advantages described in \cref{def:Cgood_c}.
In a recent work \cite{bremner2023iqp}, the authors give evidence that IQP sampling is \uvqa ~under a new conjecture.
   
\subsection{An universal efficient verifier from meta-complexity problems}

Here, we introduce variants of a meta-complexity problem for which the existence of an efficient algorithm would imply a universal polynomial-time verifier for the following class of quantum advantages. 

\begin{definition}[Sample Efficient Verifiable Quantum Advantage ($\sevqa$)]
The definition is the same as ($s,t,\epsilon$)-$\vqa$ except that the number of samples is at most $\poly(\secpar)$. 
\end{definition}

\begin{definition}[Minimum Circuit Size Problems for Samples ($\sampmcsp$)]
Let $\secpar$ be the security parameter.
Let $D$ be a distribution over $n$-bit strings and $t(\cdot)$ be any function where $n=\poly(\secpar)$. Given 
$t(\cdot)$, $s(\cdot)$, and 
polynomially many samples $z_1,\dots.z_\ell$ from $D$, the problem is to decide whether there exists a classical time-$s(n)$ sampler $S_{D'}$ such that $S_{D'}$ can sample from a distribution $D'$ such that $D$ and $D'$ are indistinguishable for any $t(n)$-time classical algorithm $\A$ with polynomially many samples:      
\begin{align*}
    |\Pr_{z_1,\dots,z_\ell\sim D}[\A(z_1,\dots,z_\ell, \sampler_D,\sampler_{D'})=1] - \Pr_{z_1,\dots,z_\ell\sim D'}[\A(z_1,\dots,z_\ell, \sampler_D,\sampler_{D'})=1]| \leq \negl(n), 
\end{align*}
where $\sampler_D$ is a sampler of $D$. 
\end{definition}

\begin{definition}[Oblivious Minimum Circuit Size Problems for Samples ($\obsampmcsp$)]
The definition is the same as above, except that the distinguisher doesn't take the description of $\sampler_{D'}$ as input.
\end{definition}

Both $\obsampmcsp$ and $\sampmcsp$ are computable. A trivial algorithm is as follows: Given samples $z_1,\dots,z_{\poly(n)}$, $s(\cdot)$, and $t(\cdot)$, the algorithm tries all $s(n)$-time samplers and $t(n)$-time distinguisher.

The following theorems show that the existence of efficient classical algorithms for $\obsampmcsp$ and $\sampmcsp$ will imply that all experiments that are $\sevqa$ or $\seuvqa$, i.e., the advantage can be verified using $\poly(n)$ samples, can be verified in classical polynomial time. In other words, algorithms for these two problems provide \emph{universal} procedures to efficiently verify $\sevqa$ or $\seuvqa$.

\begin{theorem}
\label{thm:sevqa-meta}
If $\sampmcsp$ with $(s(\cdot), t(\cdot))$ can be solved in classical polynomial time, then an $(s,t,\epsilon)$-$\sevqa$ experiment is an $(s,\mathrm{poly}(n),\epsilon+\negl(n))$-$\vqa$.     
\end{theorem}

\begin{proof}
Suppose that $\A$ is a classical polynomial-time algorithm for $\sampmcsp$. Let $z_1,\dots,z_{\poly(n)}$ be the samples generated from the experiment, $C$ be the description of the quantum circuit, and $\sampler_{D_C}$ be the description of the classical cheating sampler. Then, we can construct a polynomial-time algorithm $\A'$ to identify whether $z_1,\dots,z_{\poly(n)}$ are generated from $C$ as follows: $\A'$ on inputs $(C,\sampler_{D_C},z_1,\dots,z_{\poly(n)})$, applies $\A$ on $(z_1,\dots,z_{\poly(n)})$ and $s=|\sampler_{D_C}|$ where $|\sampler_{D_C}|$ is the circuit size of $\sampler_{D_C}$. If $\A$ outputs $1$ (i.e., there exists a classical circuit with size at most $s$), $\A'$ outputs $0$ (i.e., the samples are not from $C$); otherwise, $\A'$ outputs $1$. 

Obviously, $\A'$ runs in classical polynomial time if $\A$ is a classical polynomial-time algorithm. 

For correctness, since that quantum circuit family $\circf$ is $\sevqa$, no efficient classical sampler can generate polynomially many samples that are $t(n)$-indistinguishable from $C$ chosen randomly from $\circf$ by definition. Therefore, if $(z_1,\dots,z_{\poly(n)})$ are generated from $C$, $\A$ outputs $0$ with a probability that is at least $1-\negl(n)$. On the other hand, if $(z_1,\dots,z_{\poly(n)})$ are generated from $\sampler_D$, there exist classical samplers with size at most $|\sampler_{D_C}|$ generating samples indistinguishable from $D_C$. Therefore, $\A$ outputs $1$, and $\A'$ knows that the samples are not from $C$. This completes the proof. 
\end{proof}

The following corollary follows the same argument.
\begin{corollary}
 If $\obsampmcsp$ can be solved in classical polynomial time, then an $(s,t,\epsilon)$-$\seuvqa$ is also $(s,\mathrm{poly}(n),\epsilon+\negl(n))$-$\uvqa$.     
\end{corollary}

\section{Verifiability and $\qefid$}
\label{sec:efi}

\begin{theorem}
 Let $\secpar$ be the security parameter.
 Let $\circf = \{C_k\}$ be a set of $n$-qubit polynomial-size quantum circuits, where $n = \poly(\secpar)$. 
 The following hold:
  \begin{enumerate}[label=(\alph*)]
  \item If for all but a $\negl(\secpar)$ fraction of $C\in \circf$,
  there exists a classically poly-time samplable distribution $D_C$ such
  that $C\opro{0^n}{0^n}C^\dagger$ and $\rho_{\cdistr_C}$ form a classically secure $\qefid$ pair. Then, $\circf$
  is not a $\vqa$ family.
  \item If $\circf$ admits quantum advantage
  (i.e., for any $\distrfam=\{\cdistr_C,\sampler_C\}$,  $\E_{C \gets \circf}\norm{C\opro{0^n}{0^n}C^\dagger - \rho_{\cdistr_C}}_1 \geq 1/{\poly(\secpar)}$),  but $\circf$ is not a $\vqa$ family
  (\cref{def:weak_vqa}), then classically secure $\qefid$ exists.
  \end{enumerate}
  \label{thm:efi-vqa-duality}
\end{theorem}

\begin{proof}
    To show (a), suppose for contradiction that $\circf$ is a $\vqa$ family. Then by the characterization in~\cref{lemma:vqa=cgood_c}, $\frac{\left|\circfgood\right|}{|\circf|} \ge 1/{\poly(\secpar)}$. 
    Notice that the distinguisher used in the definition of $\circfgood$ can be used to distinguish $m$-copies of $C\ket{0^n}$
    from $m$ copies of $\rho_{\cdistr_C}$.
    By \cref{lem:efihybrid}, these $C \in \circfgood$ do not form a classically secure $\qefid$ with any $\cdistr_C$, which shows a contradiction.

    To show (b), since $\circf$ is not a $\vqa$ family, there
must exist a $\distrfam=\{\cdistr_C,\sampler_C\}$ such that for all
quantum poly-time algorithm $\A$,
\begin{align*}
  \E_{C\gets \circf}\left| \Pr[\A(C, \sampler_C, \sample_{C}) = 1] - \Pr[\A(C,
  \sampler_C, \sample_{\cdistr_C}) =
  1]\right|\le \negl(\secpar) \, , 
\end{align*}
where $\sampler_C$ is an efficient sampler for
$\cdistr_C$. 
This means for all but a $\negl(\secpar)$ fraction of $C\in \circf$, $C\ket{0^n}$ and $\rho_{\cdistr_C}$ are 
classically indistinguishable. 
Then the premise that $\circf$ admits quantum advantage implies that there is at least a $1/\poly(\secpar)$ fraction of $C\in\circf$ such that $\norm{C\opro{0^n}{0^n}C^\dagger - \rho_{\cdistr_C}}_1 \ge 1/{\poly(\secpar)}$. Hence the set of $C\in \circf$ satisfying both also constitutes at least a $1/\poly(\secpar)$ fraction, and for these $C$, $C\opro{0^n}{0^n}C^\dagger$ and $\rho_{\cdistr_C}$ form a classically secure $\qefid$. 
\end{proof}

\section{Verifiability and \prim{PRS}}
\label{sec:prs}
\begin{theorem}
    \label{thm:rcs_unverifiable}
   Let $\secpar$ be the security parameter and $n = poly(\secpar)$. Let $\circf = \{C_1,\dots, C_{N}\}$ be a set of $n$-qubit polynomial-sized quantum circuits.  If $\circf$ achieves \uvqa, then the set of states $\{C_1\ket{0^n}, \ldots, C_N\ket{0^n}\}$ is not a classically unidentifiable $\class{PRS}$.
\end{theorem}

\begin{proof}
  We prove the theorem by contrapositive. For an arbitrary distribution $D$, we let $\pmz_D: = \{\pmz_D^1,\ldots \pmz_D^m\}$  be $m$ i.i.d. samples from
  $D$, where $m=poly(\secpar)$. Suppose $\{C_1\ket{0^n}, \ldots, C_N\ket{0^n}\}$ is a classically unidentifiable $\class{PRS}$.
  For any $k \in[N]$, let $\pmz_k = \{\pmz_k^1,\ldots,\pmz_k^m\}$ be
  samples generated from measuring $\ket{\psi_k}^{\otimes m}$ with
  $\ket{\psi_k} = C_k\ket{0^n}$ in the computational basis. 
  We will describe a distribution $D$, which is efficiently samplable by a classical algorithm denoted by $\sampler_D$, such that for polynomial-time distinguisher $\A$ and a uniformly random $k\gets [N]$,
  \begin{equation}\label{eqn:zkzd}
    \left| \Pr[\A(C_k, \pmz_k)=1] - \Pr[\A(C_k, \pmz_D)=1] \right| \le
    \negl(\secpar) \, .
  \end{equation}
  This implies that $\circf$ does not achieve $\uvqa$, reaching the contradiction. 
  
    \begin{enumerate}
    \item Since the set of states
      $\{C_1\ket{0^n}, \ldots, C_N\ket{0^n}\}$ is a classically
      unidentifiable $\class{PRS}$, for any $j \neq k$ and any PPT
      algorithm $\A$, we can apply \cref{def:unidentifiable} and get
        \begin{align*}
          \abs{\Pr[\A(C_k,\pmz_k)=1]  - \Pr[\A(C_k,\pmz_j)=1]} < \negl(\secpar)\, .
        \end{align*}
      \item Let $\pmz_\mu$ be samples by measuring
        $\ket{\phi}^{\otimes m}$ in the computational basis where
        $\ket{\phi}\gets \mu$ is a Haar random state. Since
        $\{\ket{\psi_k} = C_k\ket{0} \}_{k \in [N]}$ is a $\prs$ family,
        it holds that for a uniformly random $k\gets [N]$ and any
        $j \ne k$
        \begin{align*}
          \left|\Pr[\A(C_k,\pmz_j)=1]  - \Pr[\A(C_k,\pmz_\mu)=1]
          \right| < \negl(\secpar) \, .
        \end{align*}
        Otherwise, since $C_k$ is independent of the samples, $\A( C_k, \cdot)$ can be used to build a distinguisher between
        $\prs$ and Haar random states.
           \item 
           The distribution $D$ is the uniform distribution on $\{0,1\}^n$, and 
           the sampler $\sampler_D$ simply uniformly samples $m$ distinct samples $\pmz^i_D$'s from $\{0,1\}^n$.
           By \cref{lem:conllision}, $\{\pmz^i_\mu\}_{1\leq i\leq m}$ collide with probability $1-\negl(\secpar)$. By the unitary invariance of the Haar measure, the distribution of $\{\pmz^i_\mu\}_{1\leq i\leq m}$ conditioning on no collision is uniform. Thus, the output of $\sampler_D$ is $\negl(\secpar)$-close to $\pmz_{\mu}$ with probability $1-\negl(\secpar)$.
           Thus, for all $k \in [N]$ and any PPT algorithm $\A$
           \begin{align*}
               \abs{ \Pr[\A(C_k,\pmz_D)=1] - \Pr[\A(C_k,\pmz_\mu)=1]} \leq \negl(\secpar)\, .
           \end{align*}
      \end{enumerate} 
      This theorem follows from chaining all the steps by the triangle inequality. 
\end{proof}

Set $\Delta=\{v\in\mathbb{R}^{2^n}:v_i\geq 0, \sum_i v_i=1\}$ be the set of all probability distributions on $\{0,1\}^n$. Let $(\mathbf{g}_x)_{x\in\{0,1\}^n}$ and $(\mathbf{h}_x)_{x\in\{0,1\}^n}$ be two sequences of i.i.d. random variables drawn from $N(0,1)$.
Set $\mathbf{G}=\sum_x\mathbf{g}_x^2$ and $\mathbf{H}=\sum_x\mathbf{h}_x^2$. 
Notice that $\mathbf{G}$ and $\mathbf{H}$ follow the chi-squared distribution of degree $2^n$, denoted by $\chi_{2^n}$. 
Define the random variable $\mathbf{p}$ such that for all $x \in \{0,1\}^n$
\begin{equation}\label{def:haar}
    \mathbf{p}(x)=\frac{\mathbf{g}_x^2+\mathbf{h}_x^2}{\mathbf{G}+\mathbf{H}},
\end{equation} 
which induces a probability distribution over $\Delta$. 
It is well known that $\mathbf{p}$ is the output distribution when measuring $n$-qubit Haar random states on the computational basis \cite{BS20}. 

\begin{lemma}[{\cite[comment below Lemma 1]{10.1214/aos/1015957395}}]\label{lem:chi2}
  For $n\geq 1$, let   $\mathbf{r}$ be a random variable distributed according to the chi-squared distribution $\chi_n$. Then for every $x>0$, we have
  \[
  \Pr[n-2\sqrt{nx}\leq \mathbf{r}^2\leq n+2\sqrt{nx}+2x] \geq 1-2e^{-x} \; .
  \]
\end{lemma}

\begin{lemma}\label{lem:conllision}
    Given the security parameter $\secpar, n=poly(\secpar), m=poly(n)$, let $\nu$ be a distribution drawn from $\mathbf{p}$. With probability $1-\negl(\secpar)$, the following holds.

    Let $\pmz=(\pmz_1,\ldots,\pmz_m)$ be $m$ i.i.d. samples drawn from $\nu$. Then 
    \[\Pr[\exists~i\neq j:\pmz_i=\pmz_j]\leq 50m^22^{-n}\]
\end{lemma}

\begin{proof}
Let $\nu$ be a distribution drawn from $\mathbf{p}$ defined in $\eqref{def:haar}$. Notice that $\nu_x\sim(\mathbf{g}_x^2+\mathbf{h}_x^2)/(\mathbf{G}+\mathbf{H})$.
Then the probability that samples drawn according to $v$ have a collision is
\[\Pr_{\nu}[\exists i\neq j \in [m] ~\text{s.t.}~\pmz_i=\pmz_j]\leq m^2\sum_x\nu(x)^2.\]

   Let $\mathcal{E}$ be the event that $\mathbf{G}\geq 2^n-4\sqrt{2^n n}$ and $\mathbf{H}\geq 2^n-4\sqrt{2^n n}$.
    By \cref{lem:chi2}, 
   \begin{equation}\label{eqn:Gtailbound}
       \Pr[\mathcal{E}]\geq 1-4e^{-n}.
   \end{equation}
   Then we have
    \begin{eqnarray*}
        &&\mathbb{E}_\nu[\Pr[\exists i\neq j~\text{s.t.}~\pmz_i=\pmz_j]]\\
        &\leq& m^2\mathbb{E}_\nu\left[\sum_x\frac{(\mathbf{g}_x^2+\mathbf{h}_x^2)^2}{(\mathbf{G}+\mathbf{H})^2}\right]\\
        &=& m^2\mathbb{E}_\nu\left[\sum_x\frac{(\mathbf{g}_x^2+\mathbf{h}_x^2)^2}{(\mathbf{G}+\mathbf{H})^2}\mid \mathcal{E}\right]\cdot\Pr[\mathcal{E}]+m^2\mathbb{E}_\nu\left[\sum_x\frac{(\mathbf{g}_x^2+\mathbf{h}_x^2)^2}{(\mathbf{G}+\mathbf{H})^2}\mid \neg\mathcal{E}\right]\cdot\Pr[\neg\mathcal{E}]\\
        &\leq&m^2\mathbb{E}_\nu\left[\frac{\sum_x(\mathbf{g}_x^2+\mathbf{h}_x^2)^2}{(2^n-4\sqrt{2^n n})^2}\mid\mathcal{E}\right]\cdot\Pr[\mathcal{E}]+m^2\Pr[\neg\mathcal{E}]\\
        &\leq& \frac{m^2}{2^{2n-2}}\mathbb{E}_\nu\left[\sum_x(\mathbf{g}_x^2+\mathbf{h}_x^2)^2\right]+4m^22^{-n}\\
        &=& \frac{8m^22^n}{2^{2n-2}}+4m^22^{-n}\leq 50m^22^{-n},
    \end{eqnarray*}
    where in the second inequality we use 
     \[\sum_x(\mathbf{g}_x^2+\mathbf{h}_x^2)^2\leq(\mathbf{G}+\mathbf{H})^2\]
     to bound $\mathbb{E}_\nu\left[\sum_x\frac{(\mathbf{g}_x^2+\mathbf{h}_x^2)^2}{(\mathbf{G}+\mathbf{H})^2}\mid \neg\mathcal{E}\right] \leq 1$.
    By the Markov inequality,
      
    \[\Pr_{\nu}\left[\Pr[\exists~i\neq j~\text{s.t.}~\pmz_i=\pmz_j]\leq 50m^2 2^{-n/2}\right]\geq 1-2^{-n/2}.\]
    Then the lemma follows from that $n=poly(\lambda), m=poly(n)$.     
\end{proof}
\section{Designated verifiability}
\label{sec:nizk}
Another type of verifiable quantum advantage experiments involve interactions between a trusted verifier and a computationally bounded quantum or classical prover, which is not covered by \cref{def:weak_vqa}.
Such experiments are called \emph{proof of quantumness} ($\poq$).
\begin{definition}[Proof of quantumness (($s,t,\veps$)-$\poq$)]
\label{def:ivqa}
Let $\secpar$ be the security parameter. 
We say a protocol between a classical verifier $V$ and a prover $P$ is an $(s, t,\veps)$-$\poq$, 
if there exists a quantum time-$s$ prover $P_Q$ such that
for \emph{all}
 time-$t$ classical
prover $P_C$, it holds that 
\begin{align*}
   \Pr[\langle V, P_Q\rangle = 1] - \Pr[\langle V, P_C\rangle =
  1]\geq \veps \, ,
\end{align*}
where $\langle V, P_Q\rangle$ and $\langle V, P_C\rangle$ denote the decision of $V$ after interacting with $P_Q$ and $P_C$ respectively.
\end{definition} 

Some $\poq$ protocols have been proposed in
\cite{mahadevPoQ,computationalCHSH,simplePoQ}, but only the verifier of the protocol proposed in \cite{simplePoQ} is offline.
It would be an
intriguing feature if the trusted party in $\poq$ could be offline just
as in \cref{def:weak_vqa}.  
This motivates our definition of $\vqa$
with a setup stage, where a trusted party initializes a $\vqa$
experiment with some public parameters as well as a verification key
that is issued to a designated verifier. Quantum provers can then work
offline.

\begin{definition}[Designated verifiable quantum advantages (($s,t,\veps$)-\dvqa)]
\label{def:designated_vqa}
Let $\secpar$ be the security parameter. 
Consider an experiment $E$ specified by $(\setup, P)$
where
\begin{itemize}
        \item $(\pp, \vk) \gets \setup(1^\secpar)$: a classical time-$\poly(\secpar)$ algorithm that outputs a public parameter $\pp$ and a verification key $\vk$,
        \item $z \gets P(\pp)$: a quantum time-$\poly(\secpar)$ algorithm that outputs a transcript $z$ on input $\pp$.
    \end{itemize}
We denote a classical simulation algorithm of $E$ by $\simulate$.

We say $E$ is $(s, t,\veps)$-\emph{designated verifiable quantum advantage (\dvqa)}, 
if there exists some polynomial $q$ of $\secpar$ such that for \emph{all}
 time-$t$ classical
simulator $\simulate$, there exists a classical time-$s$ algorithms  $\A(\pp, \vk, P, \simulate, \sample) \in \{0,1\}$ 
that on input $\pp$, $\vk$, the description of $P$, the description of $\simulate$, and $q(\lambda)$ transcripts generated by either $P$ or $\simulate$, outputs a bit,
such that
\begin{align*}
   \E_{(\pp,\vk) \gets \setup(1^\secpar)}\left| \Pr[\A(\pp,\vk, P, \simulate, \sample_{P}) = 1] - \Pr[\A(\pp, \vk, P, \simulate, \sample_{\simulate}) =
  1]\right| \geq \veps \, ,
\end{align*}
where $\sample_{P}$ is generated by running $P(\pp)$ $q(\secpar)$ times independently,
and $\sample_\simulate$ is generated by $\simulate$. 
\end{definition} 

It is called designated $\vqa$ because only the designated distinguisher $\A$ can get the verification key $\vk$.
When not explicitly mentioned, $s = t =\poly(\secpar)$ and $\veps = 1/\poly(\secpar)$.

Assuming a random oracle, $\poq$ of \cite{mahadevPoQ} can be made non-interactive and satisfies \cref{def:designated_vqa} \cite{alagic2020non}.
More specifically, the trusted party first generates multiple function keys along with their trapdoors. 
The function keys are published as $\pp$ and the trapdoors are kept as $\vk$.
When the prover gets $\pp$, the prover follows the steps of the original protocol of \cite{mahadevPoQ} on each function key, except that the challenge is generated by querying the random oracle.
In the end, the trusted party collects all the transcripts, runs 
the verifier's check of \cite{mahadevPoQ} on each transcript, and accepts if all of them are correct.
It is easy to see that in the new protocol, the trusted party doesn't need to stay online when the prover is generating the transcripts.
The authors of \cite{simplePoQ} cleverly design their protocol to achieve $\dvqa$ without following the Fiat-Shamir paradigm, but they still require a random oracle.

\begin{theorem}
    \label{thm:poq2dvqa} Assuming (classical) RO and LWE, there exists a $\dvqa$ experiment. 
\end{theorem}

Moreover, \cref{def:weak_vqa} can be viewed as a special case of \cref{def:designated_vqa}:
The public parameter is the circuit family $\circf$.
There is no $\vk$.
The prover $P$ runs a random $C \in \circf$ on $\ket{0^n}$ and measures the qubits in the computational basis.
The simulator runs $\sampler_{\cdistr_C}$ to generate samples.
\section{Verifying quantum advantage by a quantum verifier}
\label{sec:qvqa}

\subsection{Defining quantum verifiable quantum advantages}
\label{sec:defqvqa}

One lesson in recent developments of quantum advantage experiments is
that classically verifying the results can be challenging. Can we
employ quantum computers to help with the verification? This might
sound circular, but we think that it is a viable route worth
exploring. When we advance beyond the NISQ era, quantum advantage
experiments may be repurposed as benchmarking for quantum computers,
and checking the benchmarking metrics will be done by other quantum
computers.

In fact, we argue that it is already relevant in the NISQ era. A
classical verifier could already benefit dramatically when equipped
with limited quantum computing capacity, especially if we mindfully
tailor our experiment design to this setting. For example, interactive
protocols for proving quantumness were known relatively early as long
as a verifier can prepare some simple single-qubit
states~\cite{BFK09,ABE10}; whereas constructing a protocol with a
purely classical verifier had been notoriously challenging and was
only resolved recently in Mahadev's breakthrough
result~\cite{Mahadev22}. To put it in a real-world context, 
people are investigating whether RCS results can be verified quantumly \cite[Section V.C]{sampling_review}.
Hence, two non-colluding parties (e.g., Google vs. IBM) could verify the other
party's results, with the help of their respective NISQ device.

Hence, we extend our definitions and formalize verifiable quantum
advantage in the presence of quantum verifiers (\qvqa). On the
technical side, quantum verification enables a smoother duality
between $\efi$ and $\qvqa$, as shown in~\cref{sec:efi=qvqa}.

\begin{definition}[Quantum-verifiable quantum advantages
  ($(\cdistrs,t,\veps)$-\qvqa)]
\label{def:qvqa}
Let $\secpar$ be the security parameter. Let $\circf$ be a family of
polynomial-size quantum circuits on $n$ qubits. 

We call $\circf$ a family of $(\cdistrs, t,\veps)$
\emph{quantum-verifiable quantum advantage (\qvqa)}, if for all $\distrfam = \{\cdistr_C, \sampler_C\}_{C\in \circf}$ where $\sampler_C$ is a time-$\cdistrs$ classical sampler for $\cdistr_C$, it holds that for a uniformly random $C$ drawn from
$\circf$, there exists a time-$t$ \textbf{quantum} algorithm $\A$ that
$\veps$-distinguishes $t$ samples $\sample_{C}$ obtained from measuring $C\ket{0^n}$ from $t$ samples $\sample_{\cdistr_C}$ drawn from $\cdistr_C$, namely
\begin{align*}
  \E_{C\gets \circf}\left| \Pr[\A(C, \sampler_C, \sample_{C}) = 1] - \Pr[\A(C,
  \sampler_C, \sample_{\cdistr_C}) =
  1]\right|\geq \veps \, . 
\end{align*}
\end{definition} 

\begin{definition}[Universally quantum-verifiable quantum advantages
  ($(\cdistrs,t,\veps)$-\uqvqa)]
  \label{def:uqvqa}

  Let $\secpar$ be the security parameter. Let $\circf$ be a family of
  polynomial-size quantum circuits on $n$ qubits. 

  We call $\circf$ a family of $(\cdistrs,t,\veps)$
  \textbf{universally} quantum-verifiable quantum advantage (\uqvqa),
  if for all $\distrfam = \{\cdistr_C, \sampler_C\}_{C\in \circf}$ where $\sampler_C$ is a time-$\cdistrs$ classical sampler for $\cdistr_C$, there exists a time-$t$ quantum
  algorithm $\A$, such that for a uniformly random $C$ drawn from
  $\circf$, $\A$ $\veps$-distinguishes $t$ samples $\sample_{C}$ obtained from measuring $C\ket{0^n}$ from $t$ samples $\sample_{\cdistr_C}$ drawn from $\cdistr_C$, namely
  \begin{align*}
    \E_{C \gets \circf}\left| \Pr[\A(C, \sample_{C}) = 1] - \Pr[\A(C, \sample_{\cdistr_C}) =
    1]\right|\geq \veps \,.
  \end{align*}
\end{definition}

We are typically concerned with the efficient regime where we consider
all poly-time samplable classical distributions, poly-time
distinguisher $\A$ and inverse-poly noticeable advantage, i.e.,
$s = \poly(\secpar)$, $t = \poly(\secpar)$, and
$\veps = \frac{1}{\poly(\secpar)}$. We will simply call $\circf$ a
$\qvqa$ (resp. \uqvqa) if the conditions in~\cref{def:qvqa}
(resp.~\cref{def:uqvqa}) are satisfied in this setting.

\paragraph{Characterizing \qvqa.} Similar to what we did for \vqa, we can give an alternative characterization of \qvqa. For a circuit family $\circf$, we define $\circfgood$ to be the subset of circuits that show statistical quantum advantage and their output distributions are efficiently distinguishable from any efficiently samplable classical distributions. 

\begin{definition} 
For any $n$-qubit unitary circuit $C$, define another $2n$-qubit
unitary circuit $\hat{C}: = \cnot (C \otimes \id)$, where
$\cnot: \ket{x}\ket{y} \mapsto \ket{x}\ket{x\oplus y}$ is the
generalized CNOT gate on $n$-qubit. We define
\begin{align*}
  \rho_C : = \Tr_B(\hat C \opro{0^{2n}}{0^{2n}}_{AB}\hat C^{\dag})
  \, ,
\end{align*}
which is equivalent to a quantum state encoding the distribution induced by measuring $C\ket{0^n}$ under
the computational basis.

We say $C\in \circf$ is good, if for any distribution $\cdistr_C$ samplable by an efficient classical algorithm $\sampler_C$, the two conditions hold for some $\veps = 1/{\poly(\secpar)}$:
\begin{enumerate}[label=\roman*.]
    \item (Quantum advantage) $\norm{\rho_C - \rho_{\cdistr_C}}_1 \geq \veps $
    \item (Computationally distinguishability) $\exists$ efficient classical $\A$ s.t. 
    \begin{align*}
        \left| \Pr[\A(C, \sampler_C,
    \rho_C)  = 1] - \Pr[\A(C, \sampler_C, \rho_{\cdistr_C})
    = 1]\right| \ge \veps     
    \end{align*}
\end{enumerate}
We denote $\circfgood: = \{C\in \circf: \text{$C$ is good} \}$. 
    \label{def:cgood}
\end{definition}
Then for $\circf$ to be a $\qvqa$, $\circfgood$ has to constitute a noticeable fraction and vice versa. 

\begin{lemma} $\circf$ admits $\qvqa$ iff. $\frac{\left|\circfgood\right|}{|\circf|} \ge 1/{\poly(\secpar)}$. 
    \label{lemma:qvqa=cgood}
\end{lemma}

\begin{proof}

First the ``if'' direction follows directly from our definition of $\qvqa$. 

The ``only if'' direction follows by a hybrid argument and an averaging argument. Suppose that $\circf$ is
  a $\qvqa$ family. Then for any efficient $\distrfam = \{\cdistr_C, \sampler_C\}_{C\in \circf}$, there is a
  quantum poly-time $\A$, such that
  \begin{align*}
    \E_{C\gets \circf} \left| \Pr[\A(C, \sampler_C,
    \sample_{C}) = 1] - \Pr[\A(C, \sampler_C, \sample_{\cdistr_C})
    = 1]\right| \ge \veps \, ,
  \end{align*}
  for some $\veps \ge \frac{1}{\poly(\secpar)}$ using
  $m=\poly(\secpar)$ samples. For any $i\in[m]$, we define
  \begin{align*}
    \sample(i) &: = (\sample_C^1, \ldots, \sample_{\cdistr_C}^i,
                 \sample_{\cdistr_C}^{i+1}, \ldots, \sample_{\cdistr_C}^m) \,, \\
    \sample(i+1) &: = (\sample_C^1, \ldots, \sample_C^i,
                   \sample_{\cdistr_C}^{i+1}, \ldots, \sample_{\cdistr_C}^m)  \, .
  \end{align*}

  Then by a hybrid argument, there must
  exist an $i^*$ such that
  \begin{align*}
    \E_{C\gets \circf} \left| \Pr[\A(C, \sampler,
    \sample(i^*))  = 1] - \Pr[\A(C, \sampler, \sample(i^*+1))
    = 1]\right| \ge \veps/m \, . 
  \end{align*}
  We then construct $\A'$ to distinguish $\rho_C$ and $\rho_{\cdistr_C}$
  efficiently. On an input state $\rho \in \{\rho_C,\rho_{D_C}\}$ and
  index $i^*$\footnote{We assume $i^*$ is given to $\A'$ as a
    non-uniform advice. Alternatively, $A'$ can randomly guess $i^*$
    with, which only reduces the success probability by an
    inverse-poly factor $1/m$.}, $\A'$ constructs a state $\sigma$, 
  \begin{align*}
    \sigma: = (\rho_C^1, \ldots, \rho_C^{i-1}, \rho, \rho_{\cdistr_C}^{i^*+1},\ldots,
    \rho_{\cdistr_C}^{m}) \, 
  \end{align*}  
  and measures it in the computational basis.
  Observe that the measurement outcome is $\sample(i^*+1)$ if $\rho = \rho_C$ and
  $\sample(i^*)$ if $\rho = \rho_{\cdistr_C}$. $\A'$ then runs $\A$
  on the measurement outcomes together with $C,\sampler_C$ and outputs what $A$ outputs. We can see that
  \begin{align*}
    & \E_{C\gets \circf}\left| \Pr[\A'(\rho_C) = 1] - \Pr[\A'(\rho_{\cdistr_C})  = 1]
      \right| \\
    = & \E_{C\gets \circf}\left| \Pr[\A(C, \sampler_C,
        \sample(i^*+1)) = 1] - \Pr[\A(C, \sampler_C, \sample(i^*))  = 1] \right| \\
    \ge & \veps/m  \, .
  \end{align*}
  This implies
  that for any $\distrfam = \{\cdistr_C, \sampler_C\}_{C\in \circf}$ there must exists a $1/\poly(\secpar)$ fraction of $C\in\circf$
  such that $\A'$ successfully tells apart $\rho_C$ and $\rho_{D_C}$ with inverse-poly
  probability. These pairs also have at least $1/\poly(\secpar)$ trace distance. 
    
\end{proof}

\paragraph{Other quantum verifiable quantum advantages.} We define $(s,t,\epsilon)$-$\seqvqa$ and $(s,t,\epsilon)$-$\seuqvqa$ following the definitions of $\sevqa$ and $\seuvqa$. Briefly, they are the same as $\qvqa$ and $\uqvqa$ except that the number of samples is restricted to $\poly(n)$. Then, we can show that efficient quantum algorithms for $\sampmcsp$ and $\obsampmcsp$ can lead to polynomial-time quantum verification of $\sevqa$ and $\seuvqa$ following proofs similar to \cref{thm:sevqa-meta} as following corollaries.

\begin{corollary}
If $\sampmcsp$ with $(s(\cdot), t(\cdot))$ can be solved in quantum polynomial time, then an $(s,t,\epsilon)$-$\seqvqa$ experiment is an $(s,\mathrm{poly}(n),\epsilon+\negl(n))$-$\qvqa$. 
\end{corollary}

\begin{corollary}
 If $\obsampmcsp$ can be solved in quantum polynomial time, then an $(s,t,\epsilon)$-$\seuqvqa$ is also $(s,\mathrm{poly}(n),\epsilon+\negl(n))$-$\uqvqa$.     
\end{corollary}

 Finally, we also describe an analogue
of~\cref{def:designated_vqa} with a designated quantum verifier.

\begin{definition}[Designated quantum verifiable quantum advantages
  (($s,t,\veps$)-\dqvqa)]
\label{def:designated_qvqa}
Let $\secpar$ be the security parameter.  Consider an experiment $E$
specified by $(\setup, P)$ where
\begin{itemize}
\item $(\pp, \vk) \gets \setup(1^\secpar)$: a classical
  $\poly(\secpar)$-time algorithm that outputs a public parameter
  $\pp$ and a verification key $\vk$,
\item $z \gets P(\pp)$: a quantum $\poly(\secpar)$-time algorithm that
  outputs a transcript $z$ on input $\pp$.
    \end{itemize}
We denote a classical simulation algorithm of $E$ by $\simulate$.

We say $E$ is $(s, t,\veps)$- \emph{designated quantum verifiable
  quantum advantage (\dqvqa)}, if there exists some polynomial $q$
such that for \emph{all} time-$t$ classical simulator $\simulate$,
there exists a \emph{quantum} time-$s$ algorithms
$\A(\pp, \vk, P, \simulate, \sample) \in \bit$ such that
\begin{align*}
  \E_{(\pp,\vk) \gets \setup(1^\secpar)}\left| \Pr[\A(\pp,\vk, P, \simulate, \sample_{P}) = 1] - \Pr[\A(\pp, \vk, P, \simulate, \sample_{\simulate}) =
  1]\right| \geq \veps \, ,
\end{align*}
where $\sample_{P}$ is generated by running $P(\pp)$ $q(\secpar)$ times independently,
and $\sample_\simulate$ is generated by $\simulate$. 
\end{definition}

\subsection{Duality between $\efi$ and $\qvqa$}
\label{sec:efi=qvqa}

\begin{theorem}
  Let $\circf = \{C_k\}$ be a family of $n$-qubit poly-size quantum
  circuits. The following hold:
  \begin{enumerate}[label=(\alph*)]
  \item If for all but a $\negl(\secpar)$ fraction of $C\in \circf$,
  there exists a classically poly-time samplable distribution $D_C$ such
  that $\rho_C$ and $\rho_{\cdistr_C}$ form an $\efi$ pair. Then, $\circf$
  is not a $\qvqa$ family.
  \item If $\circf$ admits quantum advantage
  (i.e., for any $\distrfam=\{\cdistr_C,\sampler_C\}$,  $\E_{C \gets \circf}\norm{\rho_C - \rho_{\cdistr_C}}_1 \geq 1/{\poly(\secpar)}$),  but $\circf$ is not a $\qvqa$ family
  (\cref{def:qvqa}), then $\efi$ exists.
  \end{enumerate}
  \label{thm:efi-qvqa-duality}
\end{theorem}

\begin{proof}
    To show (a), suppose for contradiction that $\circf$ is a $\qvqa$ family. Then by the characterization in~\cref{lemma:qvqa=cgood}, $\frac{\left|\circfgood\right|}{|\circf|} \ge 1/{\poly(\secpar)}$. As a result, these $C \in \circfgood$ do not form a quantum-secure $\qefid$, i.e., an $\efi$, with any $\cdistr_C$, which shows a contradiction.

    To show (b), since $\circf$ is not a $\qvqa$ family, there
must exist a $\distrfam=\{\cdistr_C,\sampler_C\}$ such that for all
quantum poly-time algorithm $\A$,
\begin{align*}
  \E_{C\gets \circf}\left| \Pr[\A(C, \sampler_C, \sample_{C}) = 1] - \Pr[\A(C,
  \sampler_C, \sample_{\cdistr_C}) =
  1]\right|\le \negl(\secpar) \, , 
\end{align*}
where $\sampler_C$ is an efficient sampler for
$\cdistr_C$. Observe that $\sample_C$ is identical to multiple copies of
$\rho_C$. which is defined in \cref{def:cgood}, and $\sample_{\cdistr_C}$ is identical to multiple copies
of $\rho_{\cdistr_C}$. This means for all but a $\negl(\secpar)$ fraction of $C\in \circf$, $\rho_C$ and $\rho_{\cdistr_C}$ are indistinguishable. Then by the premise that $\circf$ admits quantum advantage,  i.e., for any $\distrfam=\{\cdistr_C,\sampler_C\}$,  $\E_{C \gets \circf}\norm{\rho_C - \rho_{\cdistr_C}}_1 \geq \veps$. This implies that there is at least a $1/\poly(\secpar)$ fraction of $C\in\circf$ such that $\norm{\rho_C - \rho_{\cdistr_C}}_1 \ge 1/{\poly(\secpar)}$. Hence, the set of $C\in \circf$ satisfying both also constitutes at least a $1/\poly(\secpar)$ fraction, and for these $C$, $\rho_C$ and $\rho_{\cdistr_C}$ form a quantum-secure $\qefid$, i.e., an $\efi$. 
\end{proof}

\bibliographystyle{alpha}
\bibliography{verify}

\end{document}